\documentclass[sigconf,nonacm]{acmart}
\usepackage{amsmath}
\DeclareMathOperator{\polylog}{polylog}
\pdfoutput=1

\AtBeginDocument{%
  \providecommand\BibTeX{{%
    \normalfont B\kern-0.5em{\scshape i\kern-0.25em b}\kern-0.8em\TeX}}}

\begin{document}
  
\title{How fast can you update your MST? \\ (Dynamic algorithms for cluster computing)}
   
\author{Seth Gilbert}
\affiliation{%
  \institution{National University of Singapore}
  \city{Singapore}
}
\email{seth.gilbert@comp.nus.edu.sg}

\author{Lawrence Li Er Lu}
\affiliation{%
  \institution{National University of Singapore}
  \country{Singapore}}
\email{lawrence.li@u.nus.edu}

\begin{abstract}
    Imagine a large graph that is being processed by a cluster of computers, e.g., described by the $k$-machine model or the Massively Parallel Computation Model.  The graph, however, is not static; instead it is receiving a constant stream of updates.  How fast can the cluster process the stream of updates?  The fundamental question we want to ask in this paper is whether we can update the graph fast enough to keep up with the stream.
    
    We focus specifically on the problem of maintaining a minimum spanning tree (MST), and we give an algorithm for the $k$-machine model that can process $O(k)$ graph updates per $O(1)$ rounds with high probability.  (And these results carry over to the Massively Parallel Computation (MPC) model.) We also show a lower bound, i.e., it is impossible to process $k^{1+\epsilon}$ updates in $O(1)$ rounds.   Thus we provide a nearly tight answer to the question of how fast a cluster can respond to a stream of graph modifications while maintaining an MST.
\end{abstract}

\maketitle



\section{Introduction}

There are two different approaches for dealing with very large graphs depending on whether the graph is static or dynamic.  

If the graph is static, it can be distributed on a cluster of machines that can then process the graph.  In the distributed algorithms world, this might be represented using the $k$-machine model~\cite{klauck15}, where the graph is randomly distributed among $k$ servers, each of which can send $\log{n}$ bits to each of the other machines in each communication round.  Alternatively, this might be represented using the Massively Parallel Computation (MPC) model~\cite{karloff10}, which is designed to capture the performance of Map-Reduce systems.  In both cases, you can efficiently find the minimum spanning tree (MST) of a graph~\cite{ghaffari_choo,klauck15}.

If the graph is dynamic, a different approach is used.  Most of the research has focused on storing as little information about the graph as possible, and treating the updates to the graph as a stream of updates. In this case, graph sketches can be used to store an approximate minimum spanning tree in $O(n \polylog(n))$ space, processing updates to the edges as they arrive, one at a time~\cite{ahn12}.

The question we ask in this paper is whether these two approaches can be combined.  How fast can a distributed cluster, e.g., a $k$-machine system, process a stream of updates to a minimum spanning tree.  Recently, Italiano et al.~\cite{italiano19} gave a first answer to this question: they showed how to maintain an \emph{approximate} MST, handling each individual update in $O(1)$ rounds.  This raised the natural question: can we maintain an exact MST, and if so, what is the fastest rate of updates that we can handle?

The main result of this paper, then, is an algorithm for maintaining an exact MST, where the graph is distributed among $k$ servers and the cluster receives a stream of updates to the graph, adding and deleting edges. Moreover, our algorithm can handle up to $\Theta(k)$ requests in $O(1)$ rounds (with high probability), allowing for significant churn in the graph if $k$ is large.  The same basic approach can be used in the MPC model.

To maintain the MST, we build on the nice idea proposed by Italiano et al.~\cite{italiano19} of using Euler tours to represent the MST. When many edges are being added to a graph, the MST may change significantly, and we develop a graph structural property that allows us to quickly determine which edges need to be added and removed from the MST. When many edges are deleted from a graph, a different approach is needed. We reduce the problem to an MST in the CONGESTED-CLIQUE model, carefully simulating the algorithm by~\cite{jurdzinski17}.  (A natural approach at simulation would take $O(k)$ time, and more care is needed in the reduction to achieve $O(1)$ time.)  Thus we can handle $O(k)$ edge insertions and deletions in only $O(1)$ rounds with high probability.

A natural follow-up question is whether it is possible to do better.  We show a matching lower bound: it is impossible to handle $k^{1+\epsilon}$ requests in $O(1)$ rounds.  Thus, if the cluster needs to keep up with the incoming stream, then it can only handle $O(k)$ updates per round without falling behind the stream of updates.  

\section{Background and Related Work} 

Large scale graphs have recently become a topic of increasing interest.
Since these large scale graphs do not fit on a single machine, the graph is stored in a distributed setting, and the algorithms are distributed in nature. However, distributed graph algorithms of the past, e.g., the CONGEST model~\cite{peleg11}, tend to treat each vertex as a single machine, while these new graphs algorithms store many vertices on a single machine, changing the nature of the distributed graph algorithms. Two models that attempt to model these large scale graphs are the $k$-machine model~\cite{klauck15}, and the popular Massively Parallel Computational(MPC) model~\cite{karloff10}. These models have been relatively well studied~\cite{pandurangan15,pandurangan18,ghaffari19, goodrich11,lattanzi11}, and similar techniques are involved in both models. 

The connectivity and MST problems have been extensively studied in the MPC and $k$-machine models.  More recently, building on work in~\cite{mohsen16,hegeman15}, Jurdzinski and Nowicki~\cite{jurdzinski17} described an algorithm that constructs an MST in $O(1)$ rounds with $O(n)$ space per machine in the MPC model.\footnote{In fact, the algorithm is presented for the CONGESTED-CLIQUE, but it implies an MPC algorithm.} It is generally believed that no $o(\log{n})$ round connectivity algorithm exists in the MPC model for the sublinear regime where each machine has $O(n^{1-\epsilon})$ space. Assadi et al.~\cite{assadi19} however, demonstrate how to obtain such round complexity when the graphs involved are sparse. 

Dynamic graph algorithms in this setting have recently come up as a topic of interest. To the best of our knowledge, Italiano et al.~\cite{italiano19} was the first paper to look at dynamic updates in the MPC model.  They introduce the dynamic MPC model, several dynamic problems in the MPC model, as well as their solutions to them. In particular, they described how Euler tours can be used to solve the dynamic connectivity and dynamic approximate MST problems in $O(1)$ communication rounds. The natural extension, batch dynamic algorithms have also been very recently studied where more than one update arrives per round.  Dhulipala et al.~\cite{dhulipala2020} build on work on the Euler Tour Tree data structure in the parallel setting~\cite{acar19,tseng2020}, and demonstrate a batch-dynamic connectivity algorithm in the MPC Model using sketching techniques.  As of yet, we know of no existing work on the dynamic \emph{exact} MST problem or the batch-dynamic MST problem.

Our work has more in common with the work by Italiano et al.~\cite{italiano19}: we generalise their results and demonstrate how Euler tours can be used to solve the dynamic MST problem in $O(1)$ communication rounds. We also demonstrate how it is, in fact, possible to resolve $k$ queries in $O(1)$ rounds with high probability, where $k$ is the number of machines. Lastly, we show that it is not possible to do much better than $k$ queries in $O(1)$ rounds, by proving a lower bound. There can be no algorithm that can resolve $k^{1+\epsilon}$ queries in $O(1)$ rounds.

While they worked in the MPC model, we will primarily describe our results for the $k$-machine model, while then describing how our approach carries over to the MPC model.

\section{Models and Problems}

In this section, we describe the main models for cluster computing, and define the MST and Dynamic MST problems.

\subsection*{$k$-Machine Model}

We focus on the \textbf{$\textbf{k}$-machine} model described by Klauck et al.~\cite{klauck15}, and highlight some of the differences between the $k$-machine model and the MPC model. We are given a graph $G = (E,V)$, with $n$ vertices and $m$ edges. 

\noindent \textbf{Graph distribution:} In the k-machine model, we assume that the graph in question $G$ is distributed across $k$ machines in the \textbf{random vertex partition} model. The $n$ vertices of the graph are distributed uniformly at random across these $k$ machines, so that each vertex has a $1/k$ chance of being on any one machine. If a vertex is distributed onto a machine, so are all of the edges it is a part of. 

\noindent\textbf{Communication:} Communications occur in synchronous rounds. The communication topology of these $k$ machines is a clique, with bidirectional links between any two machines that can only send $\log{n}$ bits a round.

\noindent\textbf{Space restrictions}: Due to the large sizes of graphs involved, we impose a space restraint on each machine. At any point in time, each machine can only use an additional $O(m/k)$ space, a constant factor amount of additional space over the space required to store the edges. Since each machine receives from up to $k$ input communication channels each round, we also assume that each machine can also use $O(k)$ space with no problems. Hence we use $O(\max\{m/k,k\})$ space.

\subsection*{MPC Model}

The \textbf{MPC} model, described by Karloff et al.~\cite{karloff10}, is usually phrased with the amount of space $S$ being an input parameter instead of the number of machines $k$. However, since these algorithms usually only use a constant factor more space over the problem size, we can in fact also think of the MPC model having the number of machines $k$ as an input parameter instead of the amount of space $S$. 

\noindent\textbf{Space restrictions:} The MPC model has $k$ machines, each with up to $S$ space, with $kS = \tilde{\Theta}(m)$, where here $\tilde{\Theta}(m)$ hides additional log factors. 

\noindent\textbf{Communication:} Communications also occur in synchronous rounds. Machines can communicate as much as they like with any other machine, as long as for each machine, the total communications in each round is $O(S)$. Contrast this with the $k$-machine model, where machines can communicate up to a total of $k$ $\log{n}$ sized messages each round, and notice that the two models scale in opposite directions: for the $k$-machine model, more machines allows for more inter-machine bandwidth; for the MPC model, more machines means less inter-machine bandwidth. 

\noindent\textbf{Graph distribution:} Since machines can exchange all of their data in one round, the graph data can be distributed arbitrarily after a single round at the beginning of the algorithm. 

For both the $k$-machine model, as well as the MPC model, we are primarily focused on minimizing the round complexity, while respecting the space constraints. 

We highlight the primary differences between the $k$-machine and the MPC models here. As we will see, Lenzen's routing lemma~\cite{lenzen13} ensures that the methods of communication are not a real difference between the two models, and that the primary difference is in the scaling of the bandwidth with respect to the number of machines. We will find that in general, algorithms in the $k$-machine model tend to work on the premise that there are a small number of machines, so that the vertex partitioning model makes more sense. Algorithms in the MPC model tend to restrict the amount of space on each machine, and we see most works focus on the amount of space required on each machine for each algorithm to work. Most work ~\cite{assadi19,ghaffari18} on the connectivity and MST problems generally work in the regime where each machine has space that is $O(n^{1 -\epsilon})$. 

Here, we also highlight the \textbf{CONGESTED CLIQUE}~\cite{lotker03} model. While not a model intended for the study of large scale graphs, we find that results in this model are particularly illuminating. The CONGESTED CLIQUE model can be thought of as the special case of the $k$-machine model where $k = n$. Instead of a random vertex partition model, we instead have a bijection between machine and vertex, and have each machine contain its vertex's edge information. The communication topology is a clique, with bidirectional links between any two machines that can only send $\log{n}$ bits a round. 

\subsection*{MST and Dynamic MST}

The \textbf{MST} problem is as follows. Given a weighted undirected graph $G$, find a spanning tree such that the total sum of the weights of the edges in this spanning tree is minimised. In the context of the $k$-machine model, we ask only that the machines know if the edges that live on their machines are in the MST or not, since storing the actual MST itself on each machine requires too much space.

The \textbf{dynamic MST} problem introduces edge additions and edge deletions. Whenever an edge is added or deleted, only the two machines this edge lives on knows about the update. We ask that the machines know if the edges that live on their machines are in the MST or not, just as in the static case.

\section{Preliminaries}

In this section, we discuss some of the basic communication primitives in the $k$-machine and MPC models.

\subsection*{Lenzen Routing} 

We begin by recalling Lenzen's routing lemma~\cite{lenzen13}. 

\begin{theorem}
The following problems can be solved in $O(1)$ communication rounds in a fully connected system of $n$ nodes:

\begin{enumerate}
	\item \textbf{Routing}: Each node is the source or the destination of up to $n$ messages of size $O(\log n)$. Only the sources know the destinations of the messages and the contents. 
	\item \textbf{Sorting}: Each node is given up to $n$ comparable keys of size $O(\log n)$. Node $i$ needs to learn the keys with indices from $(i-1)n + 1$ to $in$. 
\end{enumerate}

\end{theorem}

Lenzen's routing lemma tells us that the MPC communication model and the $k$-machine communication models are only different up to constant factors from each other, and that the only real restriction is the total bandwidth during each communication round. In the $k$-machine model, this bandwidth scales with the number of machines, while in the MPC model, this bandwidth scales inversely with the number of machines. 

\subsection*{Routing broadcasts}

A machine performs a \textbf{broadcast} if it sends the same bits through all of its communication links during that communication round. The following lemma is in the spirit of the ``Conversion Theorem'' (Theorem 4.1 of~\cite{klauck15}). While they used a randomized routing approach to obtain $O(\log{n})$ bounds, we demonstrate that a deterministic approach gives us $O(1)$ bounds.

\begin{lemma}\label{reroutinglem}
Any algorithm in the $k$-machine model that performs a total of $B$ broadcasts and/or max computations in $R$ sets, with the broadcasts and computations within each set having no dependencies, can be completed in a total of $O(B/k + R)$ rounds. 
\end{lemma}

This lemma is relatively straightforward, and its proof is available in the appendix.

\section{Dynamic MST: One at a Time}

Before going into the batch dynamic MST algorithm, we first describe the dynamic MST algorithm in this section that can handle one update a time.  In the following section, we show how to generalize this approach to $k$ updates at a time.  The main goal of this section is to prove the following theorem:

\begin{theorem}\label{mstthm}
There is an algorithm that maintains a dynamic MST in $O(1)$ communication rounds for each update.  If the graph is initially not empty, then initialization of the data structure after the MST instance has been solved takes $O(n/k)$ rounds.
\end{theorem}

We split updates into edge additions and edge deletions, handled separately. When an edge is added, we do cycle deletion to restore the MST. When an edge is deleted, we add back the minimal edge across the induced cut. As in Italiano et al.~\cite{italiano19}, where Euler tours were used to solve the dynamic connectivity and dynamic approximate MST problems, we make use of the same basic approach.  Euler tours were first used in the dynamic MST problem by Henzinger et al.~\cite{henzinger95}

\subsection{Euler tours}

An \textbf{Euler tour} in a general graph $G$ is a path that visits each edge exactly once. In the context of an MST, we treat each edge as a bidrectional edge, and an Euler tour refers to a cycle that visits each edge exactly twice. An Euler tour is the same as a depth first search edge visit order, but it is generally more useful to think of an Euler tour as a cycle. An example of an MST with an Euler tour over it can be seen in figure \ref{fig:eulertourtree}:

\begin{figure}[h]
    \centering
    \includegraphics[width=8cm]{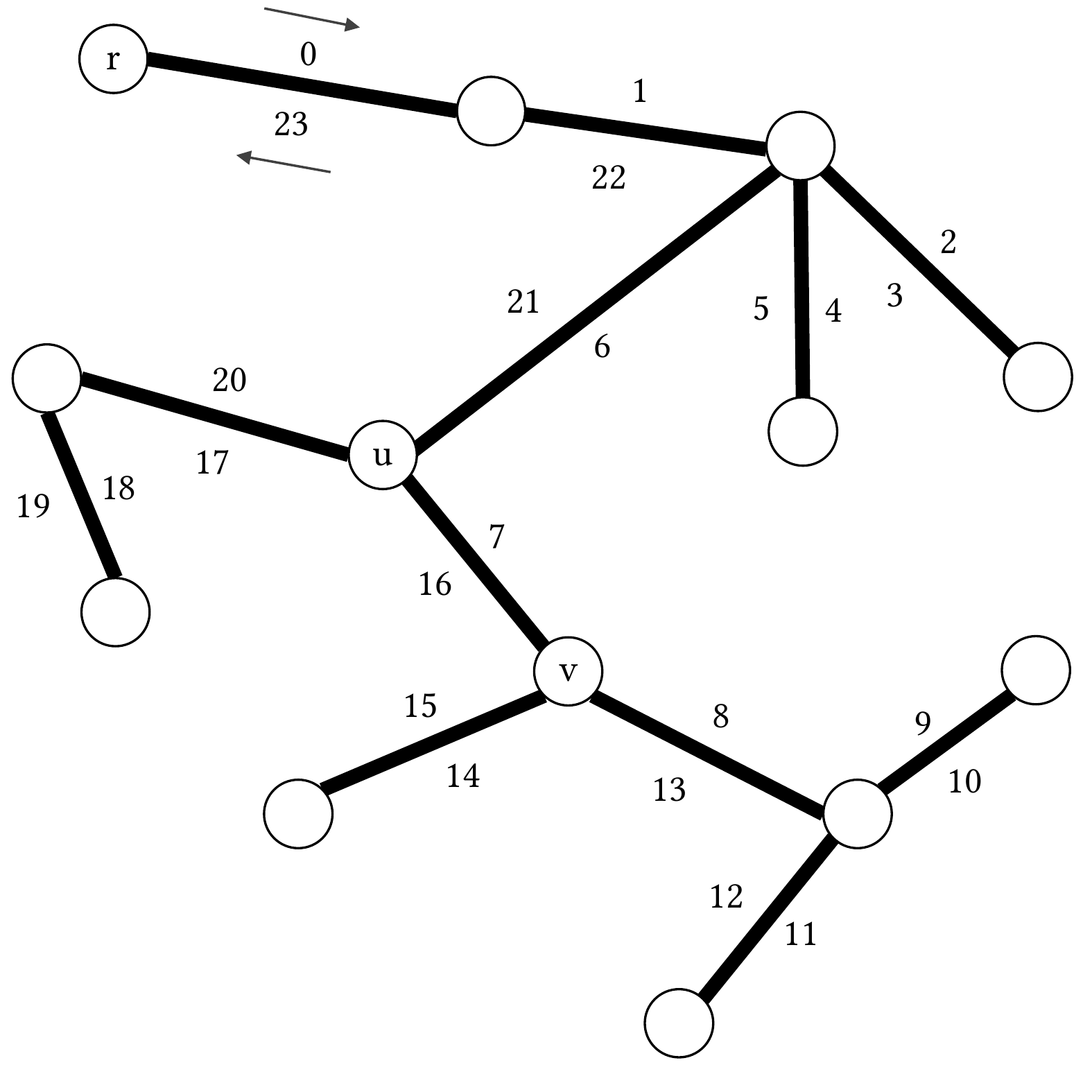}
    \caption{Euler tour over an MST, rooted at $r$}
    \label{fig:eulertourtree}
\end{figure}

We call the start of the Euler tour the \textbf{root} of the Euler tour. In general, when we refer to the root of an MST with an Euler tour structure over it, we refer to the start of the Euler tour. 

There are several different ways in which an Euler tour can be described. In Henzinger et al.~\cite{henzinger95} as well as the approach used by Italiano et al.~\cite{italiano19}, the Euler tour was described by keeping track of the order in which the vertices are visited. We employ a slightly different approach, and label the edges in the order in which they are traversed through. 

We augment each edge $e$ with these two integer values, and call the smaller one $e_{in}$, and the larger one $e_{out}$ for each of them. We now have three important lemmas:

\begin{lemma}\label{lem1}
Consider an MST $M$, rooted at $r$, and some cut edge $c \in M$ with labeled values $c_{in}$ and $c_{out}$. In the graph $M' = M\backslash c$, an edge $e$ is not in the same component as the vertex $s$ iff $e_{in} > c_{in}$ and $e_{out} < c_{out}$. 
\end{lemma}

\begin{proof} Let $M^* $ be the component separated from $s$. In the Eulerian cycle $C$, notice that $c_{in}$ denotes the time the Eulerian cycle enters the component $M^*$, and $c_{out}$ is the time it leaves the component $M^*$. As such, all the edges that are in the component $M^*$ will be visited between $c_{in}$ and $c_{out}$, and will hence have values between $c_{in}$ and $c_{out}$. 
\end{proof}

\begin{lemma}\label{lem2}
Consider an MST $M$ rooted at $r$. Consider any arbitrary vertex $s$ that is not $r$. The edge with the highest labeling with one endpoint touching $s$ and the edge with the smallest labelling with one endpoint touching $s$ are the same edge $e$.
\end{lemma}

\begin{proof} Let $e$ be the first edge that the Euler tour crosses to enter $v$. This is the desired edge $e$, since it is the first and last time the Euler tour visits the vertex $s$.\end{proof}

Let $r$ be the root of the Euler cycle, and let $s$ be any vertex, we call the edge $e$ as in lemma \ref{lem2}, the \textbf{parent edge} of $s$ with respect to $r$. In figure \ref{fig:eulertourtree}, the parent edge of $v$ with respect to $r$ is the edge $(u,v)$.

\begin{lemma}\label{lem3}
Consider an MST $M$, and some Euler tour. Let $r$ be the root of this Euler tour. Consider any arbitrary vertex $s$ that is not $r$. Let $p$ be the parent edge of $s$ with respect to $r$. An edge $e$ is on the path from $r$ to $s$ iff $e_{in} \leq p_{in}$ and $e_{out} \geq p_{out}$.
\end{lemma}

\begin{proof} Notice that an edge is on the path from $r$ to $s$ iff it is a cut edge that when removed partitions $r$ and $s$ onto two separate halves. 

($\Leftarrow$) Suppose $e_{in} < p_{in}$ and $e_{out} > p_{out}$. Then, applying lemma \ref{lem1} to the cut edge $e$ tells us that the edge $p$ is not in the same partition as $r$ with the cut edge $e$, so $e$ is a cut edge that separates $r$ and $s$ and we are done. If $e_{in} = p_{in}$ and $e_{out} = p_{out}$, then the edge $e$ is precisely the parent edge of $s$ with respect to $r$, and is the first time the component $s$ is visited, and is hence also a cut edge.

($\Rightarrow$) For the other direction, suppose $e$ is a cut edge separating $r$ and $s$. If $e$ does not touch $s$, then by lemma \ref{lem1}, the parent edge $p$ of $s$ satisfies $e_{in} < p_{in}$ and $e_{out} > p_{out}$. If $e$ does indeed touch $s$, then $e$ must be the parent edge of $s$ with respect to $r$, and we have $e_{in} = p_{in}$ and $e_{out} = p_{out}$.\end{proof}

Importantly, lemmas \ref{lem1} and \ref{lem3} give us a way to determine where edges are in the MST, from just the two values $e_{in}$ and $e_{out}$ of any edge. 

\subsection{Data structures}
To represent our Euler tour, we augment each edge in the MST with:
\begin{enumerate}
	\item The two integer values from our Euler tour, and the direction.
	\item The size of the Euler tour this edge is in.
\end{enumerate}
This additional edge information requires a constant factor more space over the original edge information.
For each machine, we also store:
\begin{enumerate}
	\item For each neighbouring vertex, the Euler tour information of a single arbitrary edge of that neighbour.
\end{enumerate}
This requires an amount of space equal to the number of neighbours, which is bounded by the number of edges on each machine, and is hence again a constant factor more space over the original edge information. 

\subsection{Maintaining the data structures}

We begin first by demonstrating several transformations that can be made in the Euler tour structures, and the number of rounds of communications required for each of them. 

\begin{lemma}\label{lem4}
Euler tours can be re-rooted after $O(1)$ broadcasts. 
\end{lemma}

\begin{proof} Suppose we wish to reroot the Euler tour to some vertex $u$. To do so, vertex $u$ broadcasts the edge value of any outgoing edge, say $d$. Each machine now subtracts $d$ from all edge values on its machines, taken modulo $2n-1$. This maintains the Euler tour structure, since Euler tours are cycles. \end{proof}

\begin{lemma}\label{lem5}
Consider an MST with an Euler tour structure over it. Given an edge $e = (u,v)$ in the MST that disconnects the MST into two separate trees, we can delete it and maintain the two separate Euler tours after $O(1)$ broadcasts.
\end{lemma}

\begin{proof} The edge being deleted broadcasts its two values $e_{min}$ and $e_{max}$. To restore the Euler tour property in both disconnected trees, we simply apply the following function $f$ to the weights $w$ globally:
$$
f(w)=
\begin{cases}
w, &\text{for } w < e_{min}\\
w - e_{min}, &\text{for } w > e_{min} \text{ and } w < e_{max}\\
w - (e_{max} - e_{min} + 1) &\text{for }w > e_{max}\\
\end{cases}
$$

Notice that this results in two Euler tours. The values in the component connected to the root have to have their values connected again, and have the large values shifted down by the number of edges removed, $e_{max} - e_{min} + 1$, while the values in the component disconnected from the root have to have their values shifted down to $0$.

We also have that the sizes of the Euler tours have to be updated. We apply the following function $g$ to the edge with weight $w$ and size $s$ globally:
$$
g(w,s)=
\begin{cases}
s - (e_{max} - e_{min} + 1), &\text{for } w < e_{min}\\
e_{max} - e_{min} - 1, &\text{for } w > e_{min} \text{ and } w < e_{max}\\
s - (e_{max} - e_{min} + 1) &\text{for }w > e_{max}\\
\end{cases}
$$

The only remaining thing to maintain is the additional Euler tour edge. In the event that machines used the edge $(u,v)$ as the edge chosen edge, since the edge was deleted, a new replacement edge is required. Here we just have both $u$ and $v$ broadcast a new edge of theirs, and we are done.\end{proof}

\begin{lemma}\label{lem6}
Consider two MSTs $M_1$ and $M_2$, both with an Euler tour structure over them. Given an edge $(u,v)$ that connects the two MSTs, we can combine the two MST and maintain the Euler tour after $O(1)$ broadcasts.
\end{lemma}

\begin{proof} The machines hosting $u$ and $v$ both broadcast the size of their individual Euler tours $s_1$ and $s_2$ respectively, as well as the value of an outgoing edge from $u$ and $v$ say $a$ and $b$ respectively. The new size of the Euler tour is then $s_1 + s_2 + 2$, and the Euler tour values are updated by the function $f_{M_1}$ and $f_{M_2}$ for the two Euler trees respectively:
\begin{align*}
f_{M_1}(w)&=
\begin{cases}
w, &\text{for } w < a\\
w + s_2 + 2 &\text{for }w >= a\\
\end{cases}\\
f_{M_2}(w) &= a + 1 + (w - b \mod s_2)
\end{align*}

The new edge $(u,v)$ has the values $a$ and $a + s_2 + 1$. This describes the Euler tour starting from $u$, passing through $(u,v)$ into $M_2$ at step $a$, and then passing back through to continue the Euler tour in $M_1$. 

Notice that no additional work is required for the additional Euler tour edge value chosen for each neighbour.\end{proof}

Lemmas \ref{lem5} and \ref{lem6} allow us to update the MST by deleting and then adding edges into the MST as required. All that remains is for us to demonstrate that the Euler tour structure allows us to determine the edges to be deleted and added when an update in $G$ occurs.

\subsection{Updates}

In this section, we describe how the edges to be added/deleted can be determined when an update arrives using the Euler tour structure in $O(1)$ communication rounds.

\subsubsection{Edge additions}

We perform edge additions as follows: We see that in the event that any edge $(u,v)$ is added, to maintain the MST, we add that edge to the MST, find the unique cycle that is created, and remove the largest weight edge from the MST. To do so, we will on any input $(u,v)$, have each machine determine if any of their edges that is in the current MST is on the path from $u$ to $v$, then a leader node will find the global maximum from the largest value from each machine. We see that each of the machines can determine if the edge is on the path from $u$ to $v$ as follows:

\begin{enumerate}
    \item Reroot the tree to $u$ using lemma \ref{lem4}.
    \item $v$ determines its parent edge $p$ and broadcasts $p_{min}$ and $p_{max}$.
    \item Edges $e$ are on the path from $u$ to $v$ iff $e_{min} < p_{min}$ and $e_{max} > p_{max}$.
    \item A max query is run on edges in this set.
\end{enumerate}

By lemma \ref{lem3}, the edges on the path from $u$ to $v$ are labeled with values such that $e_{min} \leq v_{min}$ and $e_{max} \geq v_{max}$. Now, each machine can figure out which of their edges that are in the MST have values that satisfy this property, and a leader node can figure out the global maximum. The leader node compares the current largest weight edge with the new edge, and makes the graph changes as required.

\subsubsection{Edge Deletions}\label{edgedeletionsection}

To complete edge deletions, recall that to maintain the MST after an edge $e$ in the MST is deleted, we can find the minimum weight edge across the cut and add it back into the MST. 

Lemma \ref{lem1} states that given an edge $c$ in the MST that bipartitions the graph, edge $e$ is not in the same component as the root $r$ iff $e_{in} > c_{in}$ and $e_{out} < c_{out}$. 

Let $V_i$ be the vertices that live on machine $i$, and let $N(V_i)$ be the neighbouring vertices of $V_i$ in the graph $G$. We determine the minimum edge across the cut as follows:

\begin{enumerate}
    \item The edge being deleted broadcasts the values $c_{in}$ and $c_{out}$. 
    \item For each vertex $v \in V_i\cup N(V_i)$: Pick an arbitrary edge $e$ connected to $v$, with Euler tour value $e_{in}$ and $e_{out}$. 
    \begin{itemize}
        \item If $e_{in} > c_{in}$ and $e_{out} < c_{out}$ \textbf{or} $e_{in} = c_{in}$ and $e_{in}$ is pointing away from $v$ \textbf{or} $e_{out} = c_{out}$ and $e_{out}$ is pointing towards $v$, label vertex "with root"
        \item If $e_{in} < c_{in}$ and $e_{out} > c_{out}$ \textbf{or} $e_{in} = c_{in}$ and $e_{in}$ is pointing towards $v$ \textbf{or} $e_{out} = c_{out}$ and $e_{out}$ is pointing away from $v$, label vertex "away from root
    \end{itemize}
    \item A min query is run on the edges that have endpoints with different labels.
\end{enumerate}

This is the reason why we store an additional Euler tour edge value for all neighbours, as it allows the machines to determine if edges fall on different sides of the cut.

Since each step only requires $O(1)$ broadcasts, we are done.

\subsection{Initialisation}\label{initial}

In the work by Klauck et al.~\cite{klauck15}, to demonstrate the power of their conversion theorem, they described how a Boruvka style component merging approach could allow us to construct an MST in $\tilde{O}(n/k)$ rounds. Using our Rerouting Lemma, the same approach yields the following:

\begin{theorem}
    We can construct an MST in the $k$-machine model in $O(n/k + \log{n})$ communication rounds.
\end{theorem}

The proof of this is a straightforward simulation of the Boruvka style MST algorithm using our rerouting lemma. A full proof is available in the appendix.

To complete the usage of Euler tours in our Dynamic MST problem, we demonstrate that the Euler tour structure can be initialised in the same initial round complexity. 

However, notice that a naive implementation, merging the Euler tour data structures as the components are merged is not sufficient. Our merge procedure only allows us to complete merges in pairs, but the merges required after a phase of the component merging algorithm could involve an arbitrary number of trees. The dependencies that might result would not guarantee the round complexity required. In a round where we would have to merge three components $c_1,c_2,c_3$ in a line, our previous approach would not allow us to complete this in a single round, since we would have to first merge the first two, then merge the resulting two components. \footnote{An alternate approach is to find a maximal matching of components to merge, but this is, perhaps, simpler. And it is useful later to be able to updated multiple edges in the MST at once.}

We demonstrate that we are able to merge $k$ Euler tours in $O(1)$ communication rounds.  Specifically, we prove the following lemma about $k$-way merging:

\label{sec:kwaymerge}

\begin{lemma}
Consider any forest $F$ with an Euler tour structure over each individual tree. Given a set of $k$ MST edge additions or $k$ MST edge deletions that do not create cycles, we can complete all said updates in $O(1)$ communication rounds.
\end{lemma}
\begin{proof}
Suppose these updates are ordered. (If they are not ordered, order them lexicographically.) To complete $k$ updates at once, we do the following:

\begin{enumerate}
    \item For each edge being added and deleted, we broadcast:
    \begin{itemize}
        \item An outgoing edge's Euler tour values from each endpoint.
        \item The size of the Euler tour of each endpoint.
        \item The Euler tour values of the edge if it is a deleted edge.
    \end{itemize}
    \item Each machine performs the updates in order, updating the above three values as necessary.
\end{enumerate}

Notice that at any point in time, combining two Euler tours, or separating two Euler tours only requires the above three values to be broadcast. Each machine can keep track of these values, and update them as necessary throughout the process to ensure that they are still relevant after merges and separations. Notice that outgoing edges are only involved in the edge addition case, and as such will never be deleted.

Additional work to update the the Euler tour information of neighbours only has to be completed if edges are deleted. Since at most $O(k)$ such vertices are affected, we can just broadcast them all at the end of the process.

Since each step only requires $O(k)$ broadcasts, by our rerouting lemma \ref{reroutinglem}, this can be achieved in $O(1)$ communication rounds.
\end{proof}

This establishes the procedure for initialising the Euler tour trees, and our entire algorithm is complete, and our algorithm is complete. As a result, we have proven Theorem \ref{mstthm}, the main result for this section.

Notice that this $k$-way merging lemma allows us to initialise the MST irregardless of how the MST is built. As such, independent of how the MST is determined, this process always takes $O(n/k)$ communication rounds. 

We notice here that this problem seems to lend itself well to batch updates. Updates to the tree, as well as broadcasts done to determine which edges are to be deleted or added to restore the MST only require $O(1)$ broadcasts each. This seems to suggest the possibility of resolving $O(k)$ updates in $O(1)$ communication rounds if dependencies could be avoided. 

\section{Batch Dynamic MST}

In this section, we present our main contribution: the batch dynamic minimum spanning tree algorithm. For the batch dynamic MST problem, we have $N$ updates arrive, with each update only arriving at the two machines where the updated  edge resides. The algorithm is required to determine the MST after these $N$ updates are resolved, where each machine knows which of their edges is in the MST.  The main goal of this section is to prove the following theorem:

\begin{theorem}
There is a dynamic MST algorithm in the $k$-machine model that can satisfy $k$ dynamic edge updates in $O(1)$ communication rounds, initialisation in $O(n/k + \log{n})$ rounds (if the graph is initially non-empty), while using $\max\{{k,m/k + \Delta}\}$ space, i.e., at most a constant factor more space more than the original space necessary to store the graph $G$.  The algorithm is deterministic worst case $O(1)$ in the edge addition case, and is a Las Vegas randomized algorithm for the edge deletion case, completing in $O(1)$ rounds with high probability for each attempt.
\end{theorem}

For both edge additions and deletions, our k-way updating algorithm described in Section~\ref{sec:kwaymerge} allows us to reconstruct the trees as necessary once we know which edges to add or remove.  As such, we only have to describe the procedure to determine which those edges are.

\subsection{Edge Additions}

We now prove the following lemma:

\begin{lemma}
Given a set of $k$ edge updates, we can determine the new MST in $O(1)$ communication rounds. Each machine will know if each of its edges are in the MST or not. 
\end{lemma}

When $k$ edges are added, it is not immediately clear how we can simultaneously find a set of $k$ edges to delete so that the remaining graph is both cycle free and connected. For example, if we were to pick the original $k$ cycles induced by a single new edge, as well as the existing MST edges, the maximal weight edges in all these cycles might be the same edge. Here in figure \ref{fig:kcomponentsgraph}, where the bold lines represent edges in the original MST and the dotted lines represent new edges being added, the edge labeled $(2,19)$ is in all three cycles, and might be the only edge deleted, if it were the heaviest weight edge in the graph.

It is also difficult to describe the cycles to run max queries on. Cycles could be described through the series of added edges they pass through, but such descriptions could be of length $\Theta(k)$. 

\begin{figure}[h]
    \centering
    \includegraphics[width=8cm]{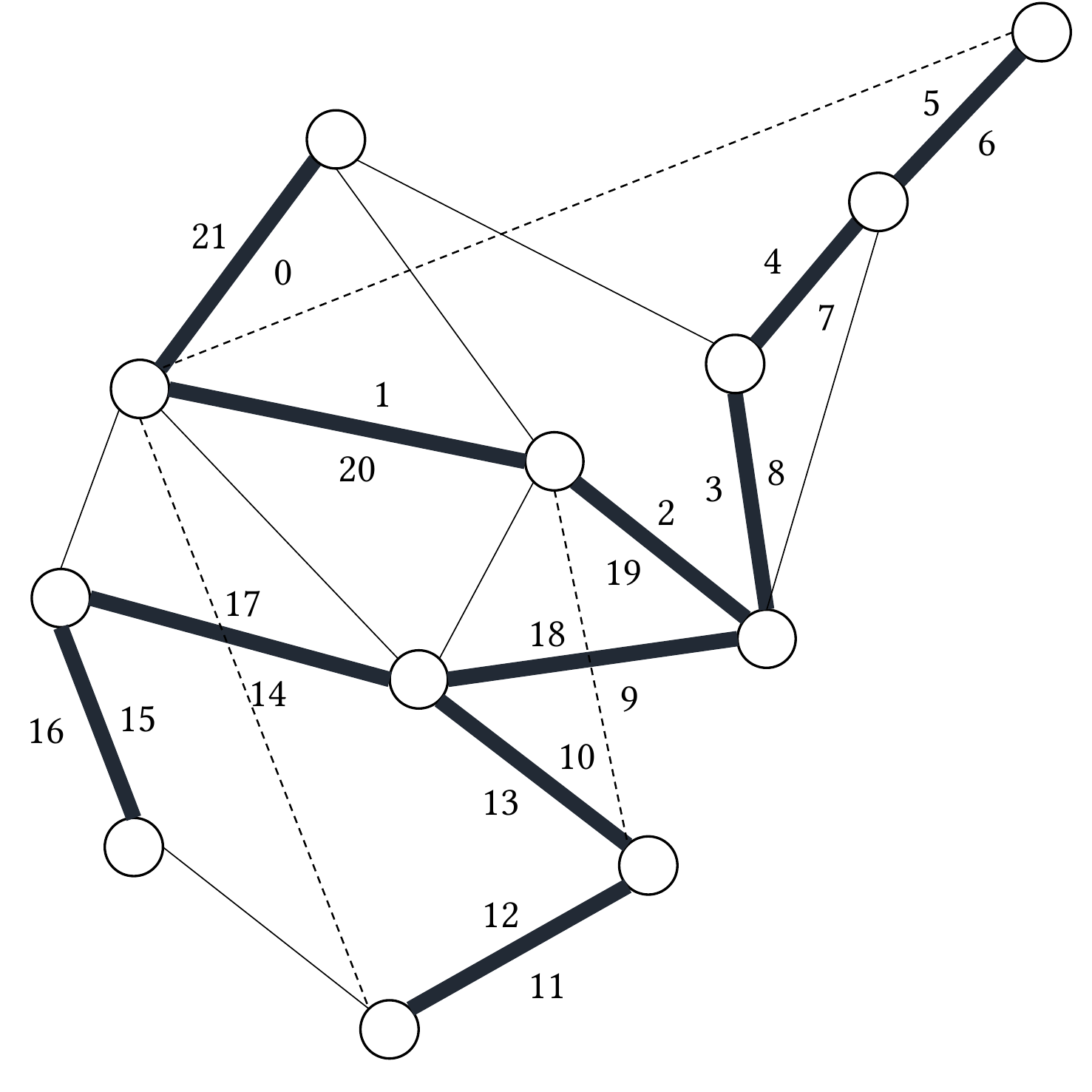}
    \caption{Example 2: Bold edges are edges in the MST, solid edges are edges in the graph $G$, dotted edges are new edges being added}
    \label{fig:kcomponentsgraph}
\end{figure}

The main insight is to notice that there are only essentially $O(k)$ edges that matter. We first begin with some intuition as to what this means. Consider again Figure~\ref{fig:kcomponentsgraph}. We first remove edges that are not in any cycles, since they are irrelevant, and can never be considered for removal. We look at the graph as if it were the original MST, with $k$ additional edges attached. In Figure \ref{fig:reducedgraph}, we can see this process in action. 

\begin{figure*}[h]
    \centering
    \includegraphics[width=\textwidth]{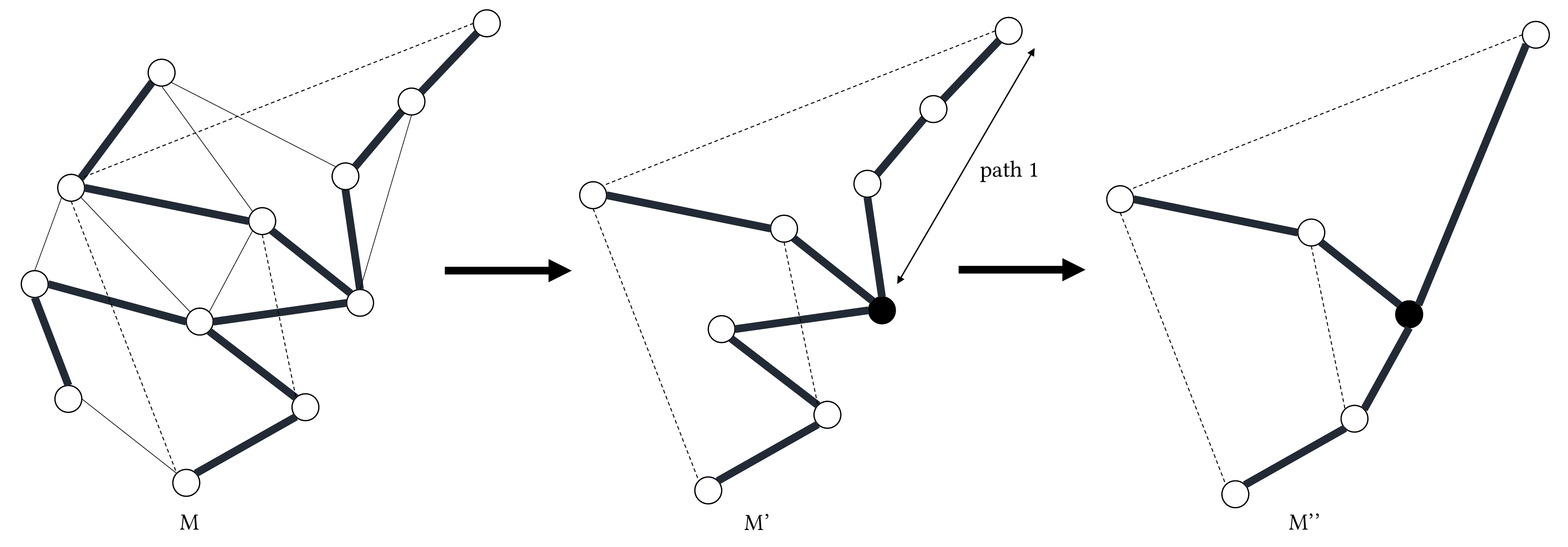}
    \caption{Example 2, removing irrelevant edges to obtain $M'$, then contracting to obtain $M''$. The shaded vertex is the sole vertex in $B$}
    \label{fig:reducedgraph}
\end{figure*}

Crucially, we wish to decompose the original MST into $O(k)$ non-intersecting paths such that at most one edge can be removed from each of the paths. As an example, refer to the decomposition of example 2 into the 5 paths described described by the third image in figure \ref{fig:reducedgraph}. Notice for example, that amongst the three edges in path 1, only one of the three edges can be deleted, if not the graph will become disconnected. After which, we can consider the contracted graph to the right, and solve the MST problem on that graph instead.

We now prove the key claim of this section:

\begin{lemma}
Given any MST, and any set of $k$ edges to connect vertices in the MST, we can decompose the edges of the MST into $O(k)$ disjoint sets such that:
\begin{itemize}
    \item At most one edge from each set can be removed while maintaining connectedness in the MST and the new edges.
    \item Each edge is in some set.
\end{itemize}
\end{lemma}
\begin{proof} 
To perform this decomposition, we first remove all edges that are not part of cycles, and place them all in one set. Call the remaining forest $M'$. We split $M'$ into paths by the following set of vertices:

\begin{itemize}
    \item Vertices that are one endpoint of the $k$ edges being added, call this set $A$.
    \item Vertices that have degree more than 2 in $M'$, for example the shaded vertex in example 2, call this set $B$.
\end{itemize}

$M'$ consists of all edges which are a part of a cycle, which are edges that are on the shortest path from some two vertices in $A$. Any leaf in this forest must be some element in $A$, if not the edge connecting to that leaf cannot be on the shortest path from some two vertices in $A$, and hence cannot be involved in a cycle.

Now, since $A$ is maximally of size $2k$, we have that the number of leaves in $M'$ is at most 2k. Since $B$ consists of the vertices that have degree two in $M'$, the number of elements in $B$ is bounded by the number of leafs in $M'$ by a degree double count. Hence, $|A| + |B| = O(k)$. Trivially the sets are disjoint.

We can now think of the induced tree $T$, with the vertices being the elements in $A$ and $B$, and the edges being the paths connecting elements in $A\cup B$ and $A \cup B$ in $M'$. Since this new graph is a forest, the number of edges it can have is bounded by the number of vertices, and is hence $O(k)$. Hence the number of sets constructed is $O(k)$. In figure \ref{fig:reducedgraph}, the graph $M''$ consists of the solid edges of $T$, and the dotted edges that are the newly added edges.

Next, notice that each of these sets is a path from some element in $A\cup B$ to some element in $A \cup B$. We wish to show that at most one edge can be removed from any such set. Suppose otherwise, and the two edges that can be removed are the edges $(w,x)$ and $(y,z)$ appearing in that order on the path. Since the remaining graph is still connected, there must be some path in $M$ from $x$ to $w$, not passing through $y$. Consider any such path, and let $s$ be the last vertex on the path from $w$ to $z$ that is visited on this path during the first time it leaves the path. Consider the partitioning of the MST induced by the edge $(x,w)$, and let $t$ be the last vertex visited in the component with $w$ on this said path, at the first time it leaves said component. 

Now, $t$ cuts across this partition, and the edge it crossed the partition with is not part of the MST, so it is one endpoint of one of the $k$ edges, and is in $A$. As such, $s$ is then either $t$, and is in $A$, and we obtain a contradiction, or the path from $s$ to $t$ consists of edges that are part of cycles, and the degree of $s$ is greater than 2, and $s$ is in $B$, also a contradiction. 

As such, each of these sets can have at most one edge removed, whilst maintaining connectivity in the original MST edges and the $k$ new edges. 

Lastly, since each edge that is part of some cycle is in these $O(k)$ sets, and all the other edges are in the first set, all edges are part of some set and we are done. 
\end{proof}

Our strategy is to run a max-query on each of these $O(k)$ sets. After which, we only have to consider these $O(k)$ edges, as well as the original new $k$ edges that are being added, and solve a contracted MST $M''$ of size $O(k)$ that can fit on a single machine. All these edges, as well as their endpoints can then be sent to all machines, and then each machine can resolve the MST on their own simultaneously. We reduce the problem to a problem on the contracted graph $M''$ with only $O(k)$ edges. 

The algorithm goes in rounds as follows:
\begin{enumerate}
    \item All the $k$ new edges being added are broadcast to all machines, so that all machines know the set $A$.
    \item Vertices in $A$ broadcast the Euler tour values of one of its edges.
    \item Vertices in $A$ determine if their edges are part of shortest paths between any elements in $A$, and broadcast all such edges.
    \item All vertices determine if they are in $B$.
    \item All vertices in $A$ and $B$ broadcast the Euler tour values of all edges connected to them that are part of a shortest path. 
    \item All machines build a picture of the tree induced, and conducts $O(k)$ max queries for the $O(k)$ sets.
    \item All of the maximums in the $O(k)$ sets are broadcast, and each machine determines the new MST, and the edges to be deleted.
    \item Euler tours are updated.
\end{enumerate}

We now describe how steps 3 and 4 work in detail. In step 3, vertices in $A$ determine if their edges are part of shortest paths between any two vertices in $A$. Notice that all such edges are on the shortest path from themselves to some other element in $A$. Hence, to determine if their edges are shortest path edges, they simulate the rerooting process, rerooting the tree to each of the other possible values in $A$, and checking to see if the edges they have are indeed parent edges with respect to some other member of $A$. The edges that are parent edges after some reroot to some element in $A$ are the edges that are on shortest paths. 

Notice too that there are only $O(k)$ such edges, since there are only $O(k)$ paths in $M'$. Hence broadcasting all these edges will take $O(1)$ communication rounds.

To determine if a vertex in $B$, it has to check that it has degree larger than $2$ in the graph induced only by shortest paths. To do so, it has to check that it has at least 3 edges connecting to it that are on shortest paths between elements in $B$.

Recall Lemma~\ref{lem3}, which states that if $r$ is the root of the Euler tour, then an edge $e$ is on the path between $r$ and $s$ iff $e_{in} < p_{in}$ and $e_{out} > p_{out}$, where $p$ is the parent edge of $s$. However, we cannot directly apply this result, since the values of $c_{in}$ and $c_{out}$ are only known to the machine hosting $s$ after rerooting the tree. To obtain all the values $c_{in}$ and $c_{out}$ would require up to $k$ broadcasts for each of them, for each other possible value of $r$ in $A$, for a total of $k^2$ broadcasts.

What is important is that the edge that is broadcast after the rerooting process is always the parent edge. As such, to avoid this problem, each machine simulates the rerooting process, and determines what the values of $c_{in}$ and $c_{out}$ would be, since they have been given all parent edges in step 3. 

To complete step 4, for each vertex $v$, for each edge connected to it, the machine checks for all the pairs of values in $A\backslash\{v\}$, and simulates the tree rerooting process, and checks to see if the edge is indeed on the shortest path between the two vertices. Now, each vertex will know its degree, and can determine if it's in $B$. 

After the sets $A$ and $B$ are determined, all vertices broadcast all the parent edges of any member of $A \cup B$ with respect to any member of $A \cup B$. Again, since there are only $O(k)$ intervals, there can only be $O(k)$ such values. 

In step 6, given the values of the Euler tour, each machine can independently build a picture of the induced tree, by placing the edges and vertices in the correct order. It then for each of the $O(k)$ sets, determines membership of the set using lemma \ref{lem3}, since it has all parent edges. It then finds the maximum weight edge in this set, and sends it to some machine for collation to find a global maximum in each set. Notice here that we can assign which machine does the collation deterministically, we simply order the paths based on the order in which they appear in the Euler tour, and take mod $k$. This results in $O(k)$ max queries that can be completed in $O(1)$ communication rounds.

After which, each machine knows precisely which edges are relevant, and can obtain the new MST. We simply delete the correct edges, with ties broken by lexicographical order, and maintain the Euler tour structure.

\subsection{Edge Deletions}

We now focus on edge deletions:

\begin{lemma}
Given a set of $k$ edge deletions, after $O(1)$ communication rounds with high probability, we can determine the new MST. Each machine will know if each of its edges are in the MST or not. 
\end{lemma}

Edge deletions only affect the MST if the deleted edges were originally in the MST. After $k$ edge deletions occur, our MST is decomposed into $k+1$ components defined by the deleted edges, and we have to find minimum edges that reconnect our components. This can be reduced to solving a new MST instance on $k$ machines and a graph with $k+1$ vertices. This is the same as solving the MST problem in the CONGESTED CLIQUE model(with the exception that we allow for $\log{n}$ bits of communication, over $\log{k}$ bits of communication). This problem has been solved very recently by Jurdziński and Nowicki in 2017~\cite{jurdzinski17} using a randomized approach. Their algorithm does not use more than $O(k)$ space.

To complete the reduction, we have to demonstrate how to convert our setting to the CONGESTED CLIQUE setting, where each machine has all the edges of its vertex. Notice here why this is not trivial. From our Euler tour data structure, we can tell for each edge, the two components it bridges. However, we cannot tell what the minimum weight edge that bridges any two components $i$ and $j$ are, without first conducting a min query that takes a broadcast. Doing $k^2$ min queries to determine all the edges in our new graph takes $O(k)$ rounds.

Here, we circumvent this problem by noticing the following fact. On each machine, there can only be $k$ edges that can possibly be in the MST. If there are more than $k$ edges that are candidates for the MST, there must be a cycle, and the machine knows that the largest weight edge on this cycle cannot be in the MST. Now, there are at most $k$ candidate edges on each machine, instead of $k^2$, and we can apply Lenzen's routing theorem. 

Notice also that we cannot apply Lenzen's routing theorem directly, since there might be more than $k$ edges that connect to a component. Multiple machines may have candidate edges that bridge some two components $i$ and $j$. The algorithm is as follows:

\begin{figure*}[h]
    \centering
    \includegraphics[width=\textwidth]{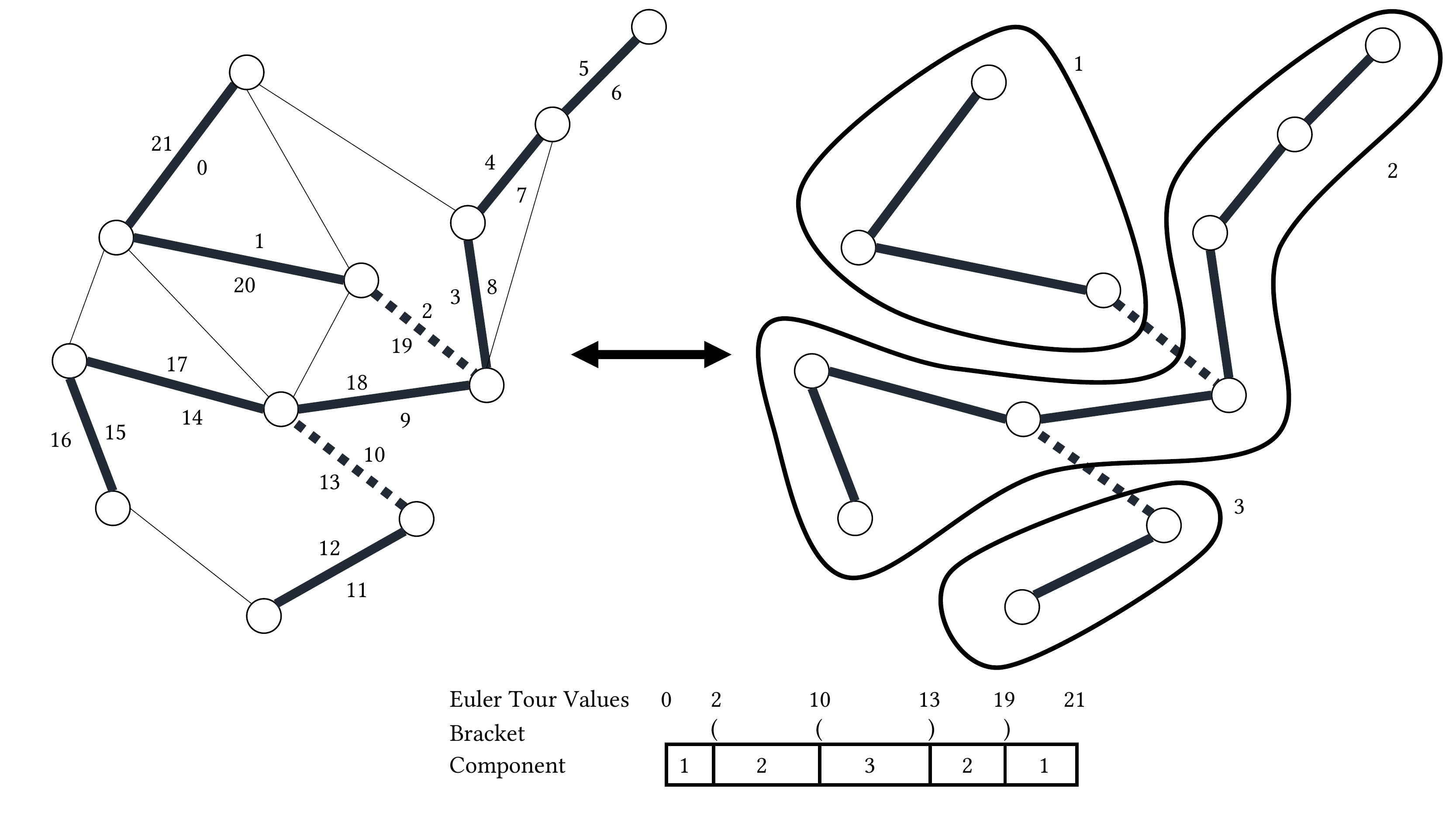}
    \caption{Example 3, Determining components using Euler Tour Values}
    \label{fig:componentsgraph}
\end{figure*}

\begin{enumerate}
	\item Broadcast all Euler tour values of $k$ edges being deleted and label disconnected components in Euler tour order.
	\item Determine which components each edge lies across. 
	\item Each machine does cycle deletion to obtain up to $k$ candidate edges.
	\item Apply Lenzen's routing theorem to sort the $k$ candidate lexicographically.
	\item Each machine keeps only the smallest weight edge across any two components.
	\item Each machine communicates with its two neighbouring machines (by index), to ensure that there are no duplicates.
	\item Use Lenzen's routing theorem to send all edges touching component $i$ to machine $i$.
	\item Run Jurdziński and Nowicki's MST algorithm.
\end{enumerate}

Steps 1 and 2 can be completed applying similar ideas to the single edge deletion case described in Section \ref{edgedeletionsection}. We construct equivalence classes as follows. Each machine first receives the Euler tour values of the $k$ edges being deleted, and lists them out in order. The smaller value of each pair of values is then represented with an open bracket, and the larger value of each pair is represented with a close bracket. All values that are contained in the same pair of brackets, and are at the same nested depth are in the same equivalence class. Each equivalence class then corresponds to the Euler tour values of a connected component. Components are labeled in order. Figure \ref{fig:reducedgraph} illustrates this process. 

Just as in Section \ref{edgedeletionsection}, we can determine which components each edge lies across with the neighbouring edge's Euler tour values. With the Euler tour values $e_{in}$ and $e_{out}$, we can determine which component the endpoint is in, by looking at where this value lies in the set of brackets determined above. In the event that the edge chosen is one of the boundary edge values(eg. 13 in figure \ref{fig:componentsgraph}), the direction of the edge is used to determine the side of the component it lies on. 

Once all the edges are labeled with the components they cut across, each machine can do cycle deletion on all of the edges $E\\E(M)$ that they have on their machines, to determine at most $k$ candidate edges that could possibly be in the new MST. 

\section{Lower Bounds}

For a lower bound, we demonstrate that it is not possible to complete much more than $O(k)$ queries in $O(1)$ rounds.

\begin{theorem}
For any constant $\delta$, there is a sequence of $3k$ batch updates, each of size $k^{1 + \delta}$, such that the total time required to complete these $3k$ batch updates is $\omega(k)$.
\end{theorem}

In~\cite{klauck15}, it was proven that the lower bound for the MST instance problem in the $k$-machine model is $\tilde{O}(n/k)$, and that the class of graphs which requires this time complexity is the following class of graphs $G_b(X,Y)$, where $X$ and $Y$ are two $b$ bit long binary string. The graph $G_b(X,Y)$ consists of $b+2$ vertices, denoted by $v_1, v_2,...v_b,u,w$. There is an edge from $u$ to $w$, and for each $ 1 \leq i \leq b$, there is an edge from $u$ to $v_i$ iff $X_i = 1$, and there is an edge from $w$ to $v_i$ iff $Y_i = 1$. There is also a guarantee that the graph is connected, and that for each $i$, $X_i \vee Y_i = 1$.

Importantly, this class of graphs has a number of edges linear in the number of vertices. 

The series of $3k$ batch updates is then as follows. We pick $k^{1 + \delta/2}$ vertices, and use the first $k$ batch updates to delete all edges that have both endpoints in this set of vertices, giving us an empty clique of size $k^{1 + \delta/2}$. The next $2k$ updates occur in pairs where we add in a random instance of the above kind, then delete it. When we add in the graph, we add it in with weights that are a global minimum. Since these new edges are all globally minimum, at the end of this batch of updates, this MST instance has to be included in the global MST, and each of these batch updates has to take $\Omega(k^{\delta/2}/\log{n})$ communication rounds, by the result in~\cite{klauck15}. This series of $k$ additions and deletions will then require at least $\Omega(k^{1 + \delta/3}/\log{n})$ communication rounds, which is $\omega(k)$. 

While for the proof in~\cite{klauck15}, $k$ was treated as a constant, and the result was with high probability in $n$, it is easy to verify that the entire proof still holds with high probability in $k$ when we set $n = k^{1 + \delta/2}$. We include a copy of the proof in the appendix. 

\section{MPC Model}

The above algorithm in the $k$-machine model maps over almost exactly to the MPC model.  The key issue to focus on is the space usage: in the $k$-machine model, we need $\Theta(m/k + \Delta)$ space on each machine. In the sublinear regime for the MPC model, we are not able to store all the edges of a high degree vertex on a single machine.  To adjust from the $k$-machine model to the MPC model, we have to shift from a vertex partitioning model to an edge partitioning model, which we can do since the MPC model allows for information to be arbitrarily reorganized.  The number of queries we can resolve also scales differently, as the communication bandwidths for the $k$-machine model and the MPC model scale differently.

\begin{theorem}
In the MPC model with $k$ machines and $S = \Theta(n^\alpha)$ space on each machine for some constant $\alpha$, such that $kS = \tilde{\Theta}(m)$, there is a dynamic MST algorithm that can satisfy $S$ dynamic edge updates in $O(1)$ communication rounds while using $S$ space, at most constant factor more space over the original space necessary to store the graph $G$. The algorithm is deterministic worst case $O(1)$ in the edge addition case, and a Las Vegas style algorithm for the edge deletion case, with it being $O(1)$ with high probability for a success in each attempt. The data structure required can be initialised in $O(\log{n})$ rounds.
\end{theorem}

We modify some parts of the $k$-machine algorithm to guarantee that the space requirements are satisfied, and follow an edge partitioning model. Each machine stores a set of edges of the graph. We however do not completely disregard the vertex partitioning model. To make it easier to complete certain vertex operations, we duplicate all edges, and store the edges on the machines lexicographically, so that any vertex is on a contiguous set of machines.

Some adjustments to the data structure have to be made to satisfy the edge partitioning model. Instead of storing the Euler tour information of a single arbitrary edge for each neighbour, we move this information onto each edge instead. For any edge $(u,v)$, we additionally store an arbitrary Euler tour edge of $u$ and an arbitrary Euler tour edge of $v$. 

A crucial difference here is in the initialisation process. Applying the initialisation argument in the $k$-machine model gives us an initialisation time of $O(n/S)$ rounds. However, MSTs in the MPC model can be solved in $O(\log{n})$ rounds in general, much faster than this initialisation time. 

To initialise the Euler tour data structure in $O(\log{n})$ rounds, we use a modified version of the Borůvka style component merging algorithm. The primary obstacle is to ensure that the edges we choose to merge do not create dependencies. Merging two components is the same as in the $k$-machine case, but in the MPC model, we can also merge stars. For any component $x$, and an arbitrary number of components connected to this component $x$ we can merge them in $O(1)$ rounds. We describe this merging process later.

To determine the stars (these do not have to include all neighbours of the central vertex) that are to be merged, we do the following. In one iteration of Borůvka's algorithm, we determine the minimum outgoing edge from each component. This set of edges forms a forest $F$. We orient the edges in this forest $F$ by orienting the edges along the minimum outgoing edge directions, with edges pointing towards each other determined by vertex id. Then, we apply the Cole-Vishkin coloring algorithm~\cite{cole-vishkin} on this oriented tree to get a 3-coloring of the forest $F$. 

Now, WLOG, let $a$ be the most frequently appearing color in this coloring of $F$. Each component colored with $a$ picks its minimum outgoing edge, and merges through this edge.

The resulting set of chosen edges cannot have any paths of length 3, and is hence a collection of stars. There are $\Theta(n)$ such edges too, resulting in $O(\log{n})$ Borůvka steps in total. What remains is to demonstrate that each round of merging can be completed in $O(1)$ rounds.

We sort the components lexicographically. For each component, we call the first machine that holds that component the leader machine for that component. From the previous step, we have obtained some collection of stars $\{S_i\}$, with the centre of each star $S_i$ being the component $s_i$. Now, each $s_i$ is unaware of the components it is supposed to merge with, but its leader node can obtain the vertices it is supposed to merge with through a converge-cast. 

Importantly here, what allows us to complete the converge-cast successfully for an arbitrary number of components, despite the leader node having only $O(S)$ bandwidth, is the nature of the Euler tour values required. We illustrate this process. Suppose some component wishes to merge with $S_1$ through the edge $(u,v)$, with $u \in S_1$. Notice that to complete the merge, each component merging with $u$ only needs to know its displacement in the Euler tour. The machine hosting the edge $(u,v)$ on the $v$ side sends the size of the component on the $v$ side to the machine hosting the edge $(u,v)$ on the $u$ side. The machine hosting the edge $(u,v)$ sums the sizes, for the converge-cast towards the leader node. After receiving the total sizes, the leader node can calculate the required displacements in the Euler tour, and send back the correct values. 


In the $k$-machine model algorithm, we extensively use broadcast and converge-cast steps. In the MPC model, it is easy to see that $O(S)$ broadcasts and converge-casts can be completed in $O(1)$ rounds using broadcast and converge-cast trees\cite{ghaffari_choo}. This is because $S = n^{\alpha}$ for some constant $\alpha$, and these trees grow by a factor of $n^\alpha$ each round, so these broadcasts take $O(1/\alpha)$ rounds. 

There are only two places in our algorithm where we make use of the fact that a single machine holds all the information about a vertex. We check that it is fine in both cases: 
\begin{itemize}
    \item For $k$-way merging, broadcasting an outgoing edge's Euler tour value from each endpoint of an added edge can be done by the leader machine for that node. 
    \item For the edge addition case, vertices verifying that they are indeed in $B$ is a simple degree check, which can be completed in a single round by sending the leader machine for that vertex the number of edges that are in $B$. 
\end{itemize}

We highlight a section of interest. In the edge deletion case, we reduce to solving an MST instance of size $S$. While solving the MST instance in $O(1)$ rounds in the sublinear regime is currently an open problem, notice that our batch size scaling to our bandwidth guarantees that we are always in the linear regime, which has been solved.


\section{Conclusion}

In this paper, we have explored how fast a cluster computing environment can maintain a minimum spanning tree subject to a sequence of updates.  Essentially, it comes down to the communication bandwidth.  In the $k$-machine model, we can handle $O(k)$ edge updates in $O(1)$ rounds.  In the MPC model where each machine has space $S$, we can handle $O(S)$ edge updates in $O(1)$ rounds.  (Of note, our contributions do not involve sketching techniques, as is common in earlier approaches, although the MST subroutine we use for the deletion case does.) We also demonstrate a lower bound for the $k$-machine model, showing that it is not possible for an algorithm to resolve $k^{1 + \epsilon}$ queries in $O(1)$ communication rounds.  One observation is that the Euler tour data structure is especially useful in the context of dynamic MST in this distributed setting. 


Future directions include expanding the approach to the problem of Steiner trees in the $k$-machine model, a structure very similar to minimum spanning trees.  Alternatively, under a different set of restrictions~\cite{pandurangan18}, it is possible to construct an MST in the $k$-machine model faster, in $O(n/k^2)$ communication rounds.  We wonder if it also possible to achieve this in the dynamic situation, obtaining $O(k^2)$ updates in $O(1)$ rounds. (In this case, only one endpoint knows that an edge is in the MST.  Surprisingly, this allows~\cite{pandurangan18} to beat the $\Omega(n/k)$ lower bound.)   We also would like to explore whether the approaches described here translate well into other distributed models. 

\subsubsection*{Acknowledgments} Thanks to Michael Bender and Martin Farach-Colton for conversations about data stream processing. Thanks to Faith Ellen for useful feedback.

\bibliographystyle{ACM-Reference-Format}
\bibliography{bibliography}

\appendix

\section{Omitted proofs}
\subsection{Rerouting Lemma}
\begin{lemma}
Any algorithm in the $k$-machine model that performs a total of $B$ broadcasts in $R$ sets, with the broadcasts within each set having no dependencies, can be completed in a total of $O(B/k + R)$ rounds. 
\end{lemma}
\begin{proof}
We can use a re-routing strategy to resolve this problem. Suppose machine $i$ has to complete $C_{i,r}$ broadcasts during round $r$. If we naively complete all the broadcasts, we will require a total of $\sum_{r = 1}^{R}\max_i\{C_{i,r}\}$ communication rounds.

Notice that in the event where one machine $i$ has to complete significantly more broadcasts than the other machines, a rerouting strategy is useful. Instead of machine $i$ broadcasting information, it instead sends $k$ sets of different information to each of the other machines, and those machines can broadcast the information instead. If $B_r = \sum_{i = 1}^kC_{i,r}$ total broadcasts are to be completed in a set, we show that this can in fact be achieved in $O(B_i/k)$ rounds. The algorithm proceeds as follows:

\begin{enumerate}
    \item Each machine broadcasts the number of broadcasts it has to do in this set to each other machine.
    \item The messages to be broadcast are globally ordered, by machine number, then by message number. Repeat the following two round procedure $B_i/k$ times. During iteration $i$:
    \begin{enumerate}
        \item Message ordered $i*k + j$ is sent to machine $j$ from the source machine.
        \item Each machine broadcasts the message it received.
    \end{enumerate}
\end{enumerate}

Notice that step 1 is essential to the success of this algorithm, since it guarantees that no two messages will be sent to the same machine in step 2a). The ordering within each machine does not need to be known by all machines, but the number of messages on other machines that have priority over it does. \end{proof}

Importantly, this strategy also applies to converge-casts. In particular, this rerouting and re-balancing strategy also works for subroutines such as a max computation, where each machine produces a value, and a global maximum is desired.

Suppose machine $i$ needs to know the maximum of these values, but is also caught up with doing several broadcasts of its own. It can reroute this max computation to any other machine $j$, and have all machines send $j$ this information instead. The process occurs as follows:

\begin{enumerate}
    \item Machine $i$ tells machine $j$ that it requires the max computation
    \item Machine $j$ broadcasts to all other machines, requesting for this information, using up the communication edges of $j$ for $O(1)$ rounds.
    \item All machines send this information to machine $j$, using up the communication edges of $J$ for another $O(1)$ rounds.
    \item Machine $j$ then sends this information back to machine $i$. 
\end{enumerate}

This completes the converge-cast. This gives us the stronger lemma:

\begin{lemma}\label{reroutinglem}
Any algorithm in the $k$-machine model that performs a total of $B$ broadcasts and/or max computations in $R$ sets, with the broadcasts and computations within each set having no dependencies, can be completed in a total of $O(B/k + R)$ rounds. 
\end{lemma}

This lemma also implies that the MST construction problem can be solved in $O(n/k + \log{n})$ rounds by simulating the Boruvka style component merging algorithm, instead of the $\tilde{O}(n/k)$ rounds as described in~\cite{klauck15}. 

\subsection{MST algorithm}
\begin{theorem}
    We can construct an MST in the $k$-machine model in $O(n/k + \log{n})$ communication rounds.
\end{theorem}

\begin{proof} 
We begin with each vertex being its own component. In each phase, we take each component and find the minimum outgoing edge, and add it to the MST, merging the two components. Finding the minimum outgoing edge is essentially a single min-query, and the merging of two components can be done in a single broadcast, to update the component names. 

After each phase, the number of components decreases by at least a factor of two, there are at most $\log{n}$ phases. The total number of min-queries across these $\log{n}$ phases is $O(n)$, since we have it bounded by $\sum_{i = 0}^{\log{n}}\frac{n}{2^i} = O(n)$ minimum outgoing edge queries. The total number of merges is $n$, so the algorithm requires a total of $O(n)$ broadcasts and min-queries, and can be completed in $O(n/k + \log{n})$ rounds applying our rerouting lemma \ref{reroutinglem}. 
\end{proof}

\subsection{Lower bound theorem proof}
Here, we replicate the proof in~\cite{klauck15} that at least $\Omega(k^{1 + \delta/3}/\log{n})$ communication rounds are required to determine an MST of with $k^{1 + \delta/2}$ vertices.

\begin{theorem}
Every public-coin $\epsilon$-error randomized protocol on a $k$-machine network, sending $\log{n}$ bits per round, that computes a spanning tree of a $k^{1 + \delta/2}$-node input graph has an expected round complexity of $\Omega(k^{\delta/2}/\log{n})$
\end{theorem}
\begin{proof} 
Let $b = k^{1 + \delta/2} - 2$. The class of graphs is  $G_b(X,Y)$, where $X$ and $Y$ are two $b$ bit long binary string. The graph $G_b(X,Y)$ consists of $b+2$ vertices, denoted by $v_1, v_2,...v_b,u,w$. There is an edge from $u$ to $w$, and for each $ 1 \leq i \leq b$, there is an edge from $u$ to $v_i$ iff $X_i = 1$, and there is an edge from $w$ to $v_i$ iff $Y_i = 1$. There is also a guarantee that the graph is connected, so that for each $i$, $X_i \vee Y_i = 1$. The total number of graphs in this class of graphs is $3^b$. 

With probability $1 - 1/k$, the vertices $u$ and $w$ are on different machines., say $p_1$ and $p_2$. To guarantee that the output is a spanning tree, the machines hosting $u$ and $w$ have to figure out which of the edges to use in the spanning tree. The proof will demonstrate that to accomplish this, there has to be a large amount of information flow. Specifically, the proof demonstrates that the conditional entropy has to change by a large amount. 

Before any communications occur, the conditional entropy $H(Y | X)$ is $2b/3$:
\begin{align*}
H(Y|X) &= \sum_xPr(X = x)\cdot H(Y | X = x)\\
&= 3^{-b}\sum_{l = 0}^b {b \choose l}2^l\cdot\log{2^l}\\
&= 3^{-b}b\sum{l = 0}^{b-1} {b - 1 \choose l} 2^{l + 1}\\
&= 2b/3
\end{align*}

Since the $k$-machine model employs the random vertex partition model, the machine hosting $u$, $p_1$ knows not only $X$, but also some vertices of $v_i$ and their edges, giving it some bits of $Y$. Let $A$ be the random variable denoting the amount of information that $p_1$ has. Employing a Chernoff bound, we can see that $p_1$ knows at most $(1 + \zeta)b/k$ bits of $Y$ for some small constant $\zeta$ with probability $1 - 2^{\zeta^2b/3k} \geq 1 - 2^{\zeta^2k^{\delta/2}/3}$. This error probability is exponentially small in $k$. In this error situation, at most $b$ bits of entropy can be lost, giving us a total reduction in entropy less than $(1 + \zeta)b/k + o(1)$. Hence, we have that $H(Y | A) \geq 2b/3 - (1 + \zeta)b/k - o(1)$. 

We now calculate the entropy at the end of the algorithm. With probability $1 - \epsilon$, the algorithm succeeds in producing a spanning tree. One of $p_1$ or $p_2$ will output at most $b/2$ edges after the algorithm ends. WLOG, let this be $p_1$. Let $E$ be the random variable of edges in the output of $p_1$, and let $T_0$ be the transcript of all messages to $p_1$. Now, we have that $H(Y | A, T_0) \leq H(Y | X, E)$, since we can simulate the algorithm and calculate both $A$ and $T_0$ from $X$ and $E$. 

We now estimate $H(Y | X, E)$. Again, we can use a Chernoff bound to obtain that $Y \leq 2b/3 + \zeta b$ with error probability exponential in $k$. Now, since $p_1$ outputs at most $b/2$ edges, at least $b/2$ edges in $Y$ have to be known from $E$. This gives us $l < 2b/3 + \zeta b - b/2 = b/6 + \zeta b$ edges that are free. These edges in $Y$ that are unknown have to correspond to edges in $X$ that have been chosen to be in the spanning tree, so there are at most $b/2$ such edges. This gives us at most
\begin{align*}
\sum_{l < b/6 + \zeta b}{b/2 \choose l} \leq b \cdot {b/2 \choose b/6 + \zeta b}
\end{align*}
possibilities for $Y$. This gives us the remaining entropy to be:
\begin{align*}
H(Y | X, E) &\leq Pr( |Y| < 2b/3 + \zeta b)(\log{{b/2 \choose b/6 + \zeta b}} + \log {b}) + o(1)\\
&\leq H(1/3 + 2\zeta)b/2 + o(b)
\end{align*}

Now, this gives us that $H(T_0 : Y | A) = H(Y | A) - H(Y|A, T_0) \geq 2b/3 - (1 + \zeta)b/k - o(1)  - H(1/3 + 2\zeta)b/2 \geq \Omega(b) - o(b)$, and hence $p_1$ has to have received messages of size $\Omega(b) = \Omega(k^{1 + \delta/2})$. Given that there are $k$ channels, with $\log{n}$ bits per channel, this gives us the desired result. 

$u$ and $w$ are on the same machine with probability $1 - 1/k$, but $(1 - 1/k)k^{1 + \delta/2} = \Omega(k^{1 + \delta/3})$, and we are done. \end{proof}

\end{document}